\documentclass[10pt,a4paper]{article}
\usepackage[english]{babel}
\usepackage{amsmath,amssymb,amsthm,bm}
\usepackage[textwidth=170mm,textheight=240mm,centering]{geometry}
\usepackage{setspace}
\usepackage{graphicx,caption}
\usepackage{float}
\usepackage[hidelinks]{hyperref}
\newtheorem{proposition}{Proposition}
\newtheorem{theorem}{Theorem}
\newtheorem{corollary}{Corollary}
\newtheorem{lemma}{Lemma}

\renewcommand\Re{\operatorname{Re}}
\renewcommand\Im{\operatorname{Im}}
\hypersetup{pdftitle={Semiclassical dynamics of higher-order squeezing: self-adjoint extensions and coherent returns},pdfauthor={Alexander E. Teretenkov},pdfkeywords={higher-order squeezing, self-adjoint extensions, Liouville operators, semiclassical limit, Husimi measures, boundary scattering}}
\title{Semiclassical dynamics of higher-order squeezing: self-adjoint extensions and coherent returns}
\author{Alexander E. Teretenkov\thanks{E-mail: \href{mailto:taemsu@mail.ru}{taemsu@mail.ru}}\\
{\small Department of Mathematical Methods for Quantum Technologies}\\
{\small Steklov Mathematical Institute of the Russian Academy of Sciences}\\
{\small Moscow, 119991 Russia}}
\date{}
\begin{document}
\maketitle

\begin{abstract}
Higher-order squeezing Hamiltonians have several self-adjoint realizations, while their classical trajectories reach infinity in finite time. We relate these extension problems through boundary amplitudes and the semiclassical propagation of coherent states. A discrete Fourier transformation gives the sink-to-source ray matrix of each quantum realization. After the first passage through infinity, the full Husimi probability measure converges to point masses on outgoing classical branches, with weights equal to the squared moduli of the matrix entries. The proof combines finite-order approximate tail solutions satisfying the exact extension boundary condition with full-norm Gaussian matching and a square-summable residual estimate. For positive energy, we determine the limit after any fixed number of passages: amplitudes evolve by powers of the unitary ray matrix, not by iterations of its squared-modulus matrix. We also realize these limits by concentrated complex amplitudes for an energy-preserving Liouville extension, distinguishing deterministic classical gluings from coherent boundary mixing.
\end{abstract}

\section{Introduction}

Higher-order generalized squeezing provides a way to prepare non-Gaussian oscillator states through the creation and annihilation of several quanta~\cite{BraunsteinMcLachlan1987}. Such states can have spatially separated phase-space components whose coherence appears in the interference structure of the Wigner function~\cite{BanaszekKnight1997}. Interactions producing higher-order squeezing have been realized in superconducting resonators~\cite{Chang2020} and trapped-ion systems~\cite{Bazavan2026}. Understanding their dynamics requires relating the phase-space distribution to classical trajectories and identifying the role of relative phases. Generalized squeezing has attracted both physical~\cite{Braeuer2025} and operator-theoretic~\cite{Fischer2026,Ashhab2026} interest. For $k>2$, the symmetric operator $\lambda(a^k+(a^\dagger)^k)$ on finite Fock sequences has deficiency indices $(k,k)$~\cite{Fischer2026}; its dynamics requires a self-adjoint realization. Even and odd truncations may select different realizations~\cite{Ashhab2026}. Singular regularization limits are studied in Ref.~\cite{FischerSingular2026}.

A limiting quantum dynamics may be described through limits of expectation values rather than limits of state vectors. In Ref.~\cite{SakbaevVolovich2017}, this approach is applied to self-adjoint regularizations of a degenerate Schr\"odinger operator with unequal deficiency indices, yielding a limiting evolution on selected observable algebras. Here the deficiency indices are equal and the semiclassical limit is taken for one fixed realization; the limiting distributions are described by Husimi measures.

The corresponding classical Hamiltonian has finite-time escape. In Refs.~\cite{Sakbaev2025,SakbaevLob2026}, V.~Zh.~Sakbaev relates continuations of incomplete flows to Liouville extensions and separates deterministic phase-point maps from general unitary amplitude transformations; the latter work develops the graph and planar Hamiltonian settings. Koopman boundary conditions in Hamiltonian statistical dynamics are discussed in Ref.~\cite{McCaul2019}. For ordinary Koopman--von Neumann transport, $|\chi|^2$ evolves autonomously, so amplitude phases do not affect later phase-space densities~\cite{Mauro2002}; an enlarged observable algebra sensitive to phases is considered in Ref.~\cite[Sec.~2]{GozziMauro2004}. In our problem, phase sensitivity comes from mixing at the boundary. Continuation by enlargement of the phase space, for singular infinite-dimensional linear Hamiltonian flows, is studied in Ref.~\cite{GlazatovSakbaev2024}. Quantum--classical comparisons using Husimi distributions at different time scales are developed in Ref.~\cite{VolovichTrushechkin2012}.

We identify the exact ray matrix of a quantum extension, prove the full weak Husimi limit through the first escape, and determine coherent returns at positive energy. The proof constructs approximate solutions of the discrete equation with exact extension boundary data; it does not infer dynamical propagation from stationary asymptotics. A Liouville amplitude representation then places the limits in the boundary-extension framework of Ref.~\cite{Sakbaev2025}. Sec.~\ref{sec:classical-flow} describes the classical extensions, Sec.~\ref{sec:WKB-correspondence} gives the quantum boundary coordinates, Sec.~\ref{sec:first-escape} proves the dynamical limits, and Sec.~\ref{sec:numerical-illustration} illustrates them.

The eikonal and amplitude-transport constructions used below belong to the standard semiclassical framework described in Ref.~\cite{MaslovFedoriuk1981}. For the canonical-operator formulation and integral representations in singular charts, see Ref.~\cite{DobrokhotovNazaikinskiiShafarevich2017}. The estimates needed for the present difference operator and its boundary condition at infinity are derived directly below.

\section{Classical geometry and Liouville extensions}
\label{sec:classical-flow}
\label{sec:liouville-extensions}

Fix $k>2$ and $\lambda>0$. The case $\lambda<0$ follows by rotating the angle by $\pi/k$. Use action--angle coordinates $I\equiv(q^2+p^2)/2$ and $\varphi\equiv\arg(q+ip)$ modulo $2\pi$, so $q=\sqrt{2I}\cos\varphi$ and $p=\sqrt{2I}\sin\varphi$. The classical Hamiltonian and its equations are
\begin{equation}
 H_k(q,p)\equiv\lambda\,2^{1-k/2}\Re(q+ip)^k
 =2\lambda I^{k/2}\cos(k\varphi),\qquad
 \dot I=\partial_\varphi H_k,\quad \dot\varphi=-\partial_I H_k.
 \label{eq:classicalHam}
\end{equation}
Here $H_k:\mathbb R^2\to\mathbb R$ is a smooth function on the phase plane with Liouville measure $dq\,dp$. The original Hamilton equations are $\dot q=\partial_pH_k$, $\dot p=-\partial_qH_k$. Equivalently, with $r\equiv\sqrt{2I}$ and $\Lambda\equiv\lambda2^{1-k/2}$,
\begin{equation}
 \dot r=-\Lambda kr^{k-1}\sin(k\varphi),\qquad
 \dot\varphi=-\Lambda kr^{k-2}\cos(k\varphi).
 \label{eq:HamEqs}
\end{equation}
The zero-energy rays are
\begin{equation}
 \varphi_l^\pm\equiv\frac{2\pi l\pm\pi/2}{k},\qquad 0\le l<k.
 \label{eq:upper-lower-rays}
\end{equation}
On the $+$ rays the motion is from infinity toward finite action (sources); on the $-$ rays it is toward infinity (sinks). A nonzero energy trajectory stays in one open sector and reaches both ends at infinity in finite time. Its travel time from the sector center to either end is
\begin{equation}
 T(c)\equiv\frac{\sqrt{\pi}}{2k^2 \Lambda^{2/k}}
 \frac{\Gamma(\frac12-\frac1k)}{\Gamma(1-\frac1k)}
 |c|^{2/k-1}.
 \label{eq:escapeTime}
\end{equation}
Indeed, on $H_k=c$ one has $r^2=|c/(\Lambda\cos(k\varphi))|^{2/k}$ and $dt=-r^2\,d\varphi/(kc)$. Hence
\begin{equation*}
 T(c)=\frac{|c|^{2/k-1}}{k^2\Lambda^{2/k}}
 \int_0^{\pi/2}(\cos u)^{-2/k}\,du,
\end{equation*}
and the beta integral gives Eq.~\eqref{eq:escapeTime}. Since $2/k-1<0$, one has $T(c)\to\infty$ as $c\to0$ and $T(c)\to0$ as $|c|\to\infty$.

We now pass from trajectories to functions on phase space. Put $\Omega\equiv\mathbb R^2\setminus\{H_k=0\}$; the removed set has Lebesgue measure zero, so $L^2(\Omega,dq\,dp)$ is naturally identified with $L^2(\mathbb R^2,dq\,dp)$. On complex-valued test functions define the densely defined skew-symmetric Liouville operator
\begin{equation}
 L_0f\equiv(\partial_pH_k)\partial_qf-(\partial_qH_k)\partial_pf,
 \qquad D(L_0)\equiv C_c^\infty(\Omega)\subset L^2(\Omega,dq\,dp).
 \label{eq:original-liouville-operator}
\end{equation}
Skew-symmetry follows by integration by parts, since the Hamiltonian vector field has zero divergence. A normalized amplitude $\chi\in L^2(\Omega,dq\,dp)$ defines a probability density $\varrho\equiv|\chi|^2\in L^1(\Omega,dq\,dp)$. The classical continuity equation and the amplitude transport equation are
\begin{equation}
 \partial_t\varrho+\partial_q(\varrho\,\partial_pH_k)
                   -\partial_p(\varrho\,\partial_qH_k)=0,
 \qquad
 \partial_t\chi+L_0\chi=0.
 \label{eq:original-liouville-equations}
\end{equation}
For smooth local solutions, differentiating $|\chi|^2$ shows that the second equation implies the first. Because trajectories escape, these equations specify only local transport until a continuation rule is chosen. The operator problem is to extend $L_0$ to a skew-self-adjoint generator, $L^*=-L$, of a global unitary amplitude evolution, as in Ref.~\cite{Sakbaev2025}.

For $c\ne0$, enumerate the $k$ components by their source ray $l$. The positive-energy component lies between $\varphi_l^-$ and $\varphi_l^+$; the negative-energy component lies between $\varphi_l^+$ and $\varphi_{l+1}^-$, with indices modulo $k$. Energy and time along the orbit, with $s=0$ at the center, identify each component with $-T(c)<s<T(c)$. In these coordinates $\dot c=0$, $\dot s=1$ and $dq\,dp=dc\,ds$: the latter identity follows from the unit absolute Jacobian of energy--time coordinates, as in Ref.~\cite[Lemma~18]{Sakbaev2025}.

\begin{proposition}[Energy-fiber preserving extensions]
\label{prop:energy-fiber-extensions}
For the inverse energy--time coordinates $(q_l,p_l)$ on component $l$, pullback $\widetilde\chi_l(c,s)\equiv\chi(q_l(c,s),p_l(c,s))$ identifies $L^2(\Omega,dq\,dp)$ unitarily with
\begin{equation*}
 \int_{\mathbb R\setminus\{0\}}^\oplus
 L^2((-T(c),T(c));\mathbb C^k)\,dc.
\end{equation*}
It carries the operator in Eq.~\eqref{eq:original-liouville-operator} to $\partial_s$ and the amplitude equation in Eq.~\eqref{eq:original-liouville-equations} to $\partial_t\widetilde\chi+\partial_s\widetilde\chi=0$. In the target space we suppress the tilde on amplitudes. The maximal Liouville expression is $\widetilde L_{\max}\equiv\partial_s$, with distributional domain $\{\chi:\chi,\partial_s\chi\in L^2\}$ in this direct-integral space. Set
\begin{equation*}
 \tau_\pm\chi(c)\equiv(\chi_l(c,\pm T(c)))_{l=0}^{k-1}.
\end{equation*}
Every measurable $J(c)\in U(k)$ defines a skew-self-adjoint restriction
\begin{equation}
 D(\widetilde L_J)\equiv
 \{\chi\in D(\widetilde L_{\max}):\tau_+\chi(c)=J(c)\tau_-\chi(c)\text{ a.e.}\}.
 \label{eq:classical-extension-domain}
\end{equation}
Conversely, these are all extensions whose unitary groups strongly commute with multiplication by $c$.
\end{proposition}

\begin{proof}
The measure identity gives unitarity. Since $\partial_sq_l=\partial_pH_k$ and $\partial_sp_l=-\partial_qH_k$, the chain rule transforms $L_0$ into $\partial_s$, as in Ref.~\cite[Lemma~6]{SakbaevLob2026}. For almost every $c$, maximal-domain components belong to $H^1(-T(c),T(c))$. Integration by parts gives the boundary form $\langle\tau_+u,\tau_+v\rangle-\langle\tau_-u,\tau_-v\rangle$; its unitary graphs classify the skew-self-adjoint derivatives. At each fixed energy this is the unit-speed graph construction of Ref.~\cite[Theorem~1]{SakbaevLob2026}. The direct integral gives Eq.~\eqref{eq:classical-extension-domain}; strong commutation with multiplication by $c$ gives the converse by decomposability. This is the fiberwise version of Ref.~\cite[Theorem~19]{Sakbaev2025}. Traces are used fiberwise only: unweighted global $L^2(dc)$ traces need not exist when $T(c)\to0$.
\end{proof}

A channel means one connected nonzero-energy trajectory component, with its source at $s=-T(c)$ and sink at $s=T(c)$. A deterministic, energy- and measure-preserving continuation glues each sink endpoint to one source endpoint, hence is a measurable family of channel permutations. Conversely, any such family defines a continued measure-preserving flow. The observable boundary condition is its inverse permutation, because observables transform by pullback; this specializes Ref.~\cite[Theorems~12 and 20]{Sakbaev2025}. For unit-speed graph transport, permutation boundary matrices are characterized by positivity preservation~\cite[Theorem~5]{SakbaevLob2026}. In the planar criterion of Ref.~\cite[Theorem~7]{SakbaevLob2026}, an energy-preserving Koopman continuation requires both preservation of function value ranges and commutation with multiplication by energy. General unitary $J(c)$ need not describe a phase-point flow: energy-fiber preservation alone is insufficient. It concerns the classical observable $c$, not the quadratic form of $i\widetilde L_J$.

\section{Quantum boundary coordinates and ray scattering}
\label{sec:WKB-correspondence}

We denote the identity matrix on $\mathbb C^d$ by $\mathbf 1_d$, reserving $I$ for the classical action. Let $[Q,P]=i\hbar$, $a_\hbar\equiv(Q+iP)/\sqrt{2\hbar}$. Weyl quantization of Eq.~\eqref{eq:classicalHam} gives
\begin{equation*}
 \widehat H_{k,\hbar}\equiv\lambda\hbar^{k/2}
 (a_\hbar^k+(a_\hbar^\dagger)^k).
\end{equation*}
The operator is initially defined on finite Fock sequences. There is no ordering correction, since $(q\pm ip)^k$ are powers of linear functions. Occupation number modulo $k$ decomposes the Fock space into $\mathcal H_r\equiv\overline{\operatorname{span}}\{|r+km\rangle:m\ge0\}$, $0\le r<k$. The sector Jacobi expression is
\begin{equation}
 (J_{r,\hbar}u)_m\equiv b_{m,\hbar}^{(r)}u_{m+1}+b_{m-1,\hbar}^{(r)}u_{m-1},\qquad
 b_{m,\hbar}^{(r)}\equiv\lambda\hbar^{k/2}
 \sqrt{\frac{(r+(m+1)k)!}{(r+km)!}},
 \label{eq:Jacobi-coefficients-WKB}
\end{equation}
with the lower term absent at $m=0$. A tail solution satisfies $(J_{r,\hbar}-z)u=0$ for $m\ge1$; the full eigenvalue equation also requires $b_{0,\hbar}^{(r)}u_1=zu_0$.

Tridiagonal operator realizations and their unitary groups are also studied in Ref.~\cite{Teretenkov2026Diagonalizations}. In the $\mathfrak{su}(1,1)$ representations considered there, the off-diagonal coefficients grow linearly and satisfy the Carleman condition, ensuring essential self-adjointness on finite sequences.

\begin{theorem}[Asymptotic boundary coordinates]
\label{thm:WKB-boundary-coordinates}
For every $z\in\mathbb C$ the tail recurrence has independent solutions
\begin{equation}
 w_{r,m}^{\pm}(z)=
 (b_{m,\hbar}^{(r)})^{-1/2}e^{\pm i\pi m/2}(1+o(1)).
 \label{eq:WKB-solutions}
\end{equation}
Every maximal-domain vector has unique coefficients $s_{r,\pm}$ such that
\begin{equation}
 (\psi_r)_m=(b_{m,\hbar}^{(r)})^{-1/2}
 [s_{r,+}e^{i\pi m/2}+s_{r,-}e^{-i\pi m/2}+o(1)].
 \label{eq:maximal-domain-asymptotics}
\end{equation}
The maps
\begin{equation}
 \Gamma_{\pm,\hbar}\psi\equiv\sqrt2\,(s_{r,\pm}(\psi_r))_{r=0}^{k-1}
 \label{eq:boundary-maps}
\end{equation}
jointly map the maximal domain onto $\mathbb C^k\oplus\mathbb C^k$ and satisfy, with inner products linear in the first argument,
\begin{equation}
 \langle\widehat H_{k,\hbar}^*\psi,\chi\rangle-
 \langle\psi,\widehat H_{k,\hbar}^*\chi\rangle
 =i(\langle\Gamma_{+,\hbar}\psi,\Gamma_{+,\hbar}\chi\rangle-
       \langle\Gamma_{-,\hbar}\psi,\Gamma_{-,\hbar}\chi\rangle).
 \label{eq:WKB-Green-identity}
\end{equation}
The deficiency indices are $(k,k)$, and all self-adjoint realizations are
\begin{equation}
 D(\widehat H_{k,\hbar}^{U})\equiv
 \{\psi\in D(\widehat H_{k,\hbar}^*):
       \Gamma_{+,\hbar}\psi=U\Gamma_{-,\hbar}\psi\},\qquad U\in U(k).
 \label{eq:WKB-extension-domain}
\end{equation}
\end{theorem}

\begin{proof}
Apply Ref.~\cite[Theorems~3.1, 3.6 and Proposition~3.7]{Yafaev2021} in each sector. Indeed, $b_m\sim\lambda(\hbar k)^{k/2}m^{k/2}$, $\sum b_m^{-1}<\infty$, the normalized diagonal is zero, $b_m/\sqrt{b_{m-1}b_{m+1}}=1+O(m^{-2})$, and $\sqrt{b_{m+1}/b_m}\to1$. These verify conditions (1.11) and (3.7)--(3.9) of that reference; its phase (3.10) is $\pi m/2$. The two tail solutions are square summable, not necessarily solutions of the left-endpoint equation. The cited maximal-domain asymptotics and Green identity give Eqs.~\eqref{eq:maximal-domain-asymptotics}--\eqref{eq:WKB-Green-identity}. Unitary graphs of the boundary form give Eq.~\eqref{eq:WKB-extension-domain}, consistently with Ref.~\cite[Example~5.2 and Theorem~4.16]{Fischer2026}.
\end{proof}

Since $b_{m,\hbar}^{(r)}=\hbar^{k/2}b_{m,1}^{(r)}$, identification through the same Fock basis gives $\Gamma_{\pm,\hbar}=\hbar^{k/4}\Gamma_{\pm,1}$. Thus fixed $U$ defines the same boundary relation for every $\hbar$. The matrix $U$ parametrizes the extensions in the boundary coordinates $\Gamma_{\pm,\hbar}$ defined by Eq.~\eqref{eq:boundary-maps}; transition to a von Neumann parameter between deficiency spaces is generally fractional-linear, not an identification of numerical matrices.

At fixed energy $c$, the discrete relation $2b_m\cos\theta=c$ gives two phase branches tending to $\pm\pi/2$ as $m\to\infty$; this is the local Liouville--Green identification, cf. Ref.~\cite{Yafaev2022Spectral}. Its geometric sign follows directly from
\begin{equation*}
 \langle q,p\mid\psi\rangle=e^{-|\xi|^2/2}
 \sum_{n\ge0}\frac{\bar\xi^{\,n}}{\sqrt{n!}}\psi_n,\qquad
 \xi\equiv\frac{q+ip}{\sqrt{2\hbar}}.
\end{equation*}
The factor $e^{-i(r+km)\varphi}$ makes $\theta=k\varphi$, so $i^m$ corresponds to source rays $\varphi_l^+$ and $(-i)^m$ to sink rays $\varphi_l^-$. The ray transforms and exact boundary law are
\begin{equation}
 (F_-)_{lr}\equiv k^{-1/2}e^{-ir\varphi_l^+},\qquad
 (F_+)_{lr}\equiv k^{-1/2}e^{-ir\varphi_l^-},\qquad
 A_{\rm src}\equiv F_-\Gamma_{+,\hbar}\psi,\quad
 A_{\rm snk}\equiv F_+\Gamma_{-,\hbar}\psi,
 \label{eq:F-plus-minus}
\end{equation}
\begin{equation}
 A_{\rm src}=S_{\rm q}(U)A_{\rm snk},\qquad
 S_{\rm q}(U)\equiv F_-UF_+^*.
 \label{eq:channel-boundary-condition}
\end{equation}
Both $F_\pm$ are unitary and map residue-sector amplitudes to individual ray amplitudes. Changing the two ray bases multiplies $S_{\rm q}$ on the left and right by diagonal unitary matrices. In particular, normalized sink amplitudes tending to $e_j$ have source probabilities tending to $|(S_{\rm q}(U))_{lj}|^2$. This algebraic statement alone is not a propagation theorem.

\begin{proposition}[Deterministic gluings and truncations]
\label{prop:quantized-deterministic-gluings}
The quantum extension has an underlying deterministic sink-to-source gluing $g\in S_k$ precisely when
\begin{equation}
 S_{\rm q}(U)=DP_g,\qquad U=F_-^*DP_gF_+,
 \qquad (P_gz)_l\equiv z_{g^{-1}(l)},
 \label{eq:deterministic-extension}
\end{equation}
where $D$ is diagonal unitary. Any family represented by this fixed matrix, up to diagonal phases, has the same $g$ almost everywhere in energy. For $g_s(j)\equiv j+s\pmod k$ and $D=\mathbf 1_k$,
\begin{equation}
 U_s\equiv\operatorname{diag}(e^{i\pi(2s+1)r/k})_{r=0}^{k-1}.
 \label{eq:cyclic-extension-matrix}
\end{equation}
At fixed $\hbar$, truncation to the first $N$ sites in each sector, extended by zero, converges in strong resolvent sense to $U=-\mathbf 1_k$ for even $N$ and $U=\mathbf 1_k$ for odd $N$; the unitary groups converge strongly at every fixed time. Independent sector parities give the corresponding diagonal signs.
\end{proposition}

\begin{proof}
A unitary matrix takes every channel into a single channel exactly when it is monomial, giving Eq.~\eqref{eq:deterministic-extension}; its nonzero pattern uniquely determines $g$. Writing $F_\pm=F_0^*\operatorname{diag}(e^{\pm i\pi r/(2k)})$ with $(F_0)_{lr}\equiv k^{-1/2}e^{2\pi ilr/k}$ gives Eq.~\eqref{eq:cyclic-extension-matrix}. For truncations, use Ref.~\cite[Proposition~IV.1 and Theorem~IV.2]{Ashhab2026}. In our positive Jacobi convention, the zero-energy tail solution supported on odd sites has $s_+=-s_-$ and vanishes at the omitted index for even $N$, selecting $-1$; the even-supported solution selects $+1$ for odd $N$. Finite sequences together with this tail vector form a core of the selected extension. For each vector in this core, its truncations converge in Hilbert-space norm, and their images under the truncated operators converge to its image under the selected extension. This gives strong resolvent and unitary-group convergence as in the cited proof.
\end{proof}

For a general matrix the one-pass probability rule is $\mathsf P(U)_{lj}\equiv|(S_{\rm q}(U))_{lj}|^2$, a bistochastic matrix. An incoherent channel mixture evolves by $\mathsf P(U)$, but this discards relative phases; the next section proves the coherent dynamical law instead.

\section{Semiclassical passage through infinity}
\label{sec:first-escape}

To construct packets through infinity, we rewrite the Jacobi recurrence in the classical action variable $I$. The oscillator action operator is $(Q^2+P^2)/2=\hbar(a_\hbar^\dagger a_\hbar+1/2)$, so the $n$th Fock state has action $\hbar(n+1/2)$. In residue sector $n=r+km$, these values form an equally spaced lattice. Put $\alpha\equiv k/2>1$ and define its nodes, step and coefficient interpolation by
\begin{equation}
 \begin{gathered}
 x_n\equiv\hbar(n+\tfrac12),\qquad
 x_{r,m}\equiv x_{r+km},\qquad \delta\equiv k\hbar,\\
 A_\hbar(x)\equiv\lambda\prod_{\ell=1}^{k}
 \bigl(x+\hbar(\ell-\tfrac{k+1}{2})\bigr)^{1/2}.
 \end{gathered}
 \label{eq:midpoint-action}
\end{equation}
Here $x$ is a continuous action variable. The function $A_\hbar$ is a smooth interpolation of the off-diagonal coefficients: the coefficient coupling the nodes $x_{r,m}$ and $x_{r,m+1}$ is its value at their midpoint. This is an explicit construction for our recurrence, since substitution into Eq.~\eqref{eq:Jacobi-coefficients-WKB} gives
\begin{equation}
 b_{m,\hbar}^{(r)}=A_\hbar(x_{r,m}+\delta/2),\qquad
 A_\hbar(x)=\lambda x^\alpha\bigl(1+O(\hbar^2/x^2)\bigr).
 \label{eq:midpoint-jacobi}
\end{equation}
The square roots are positive for $x>\hbar(k-1)/2$, in particular on the tail used below. Symmetry of the shifts $\ell-(k+1)/2$ makes $A_\hbar$ even in $\hbar$ and gives the expansion in Eq.~\eqref{eq:midpoint-jacobi}, with no linear term. The same interpolation works for every residue $r$.

Next we choose the oscillatory phase for a packet of energy $c$. Substituting the leading ansatz $e^{iS(x,c)/\hbar}$ into the recurrence gives $2\lambda x^\alpha\cos(kS_x)=c$. This is the discrete eikonal relation; for background on such Liouville--Green constructions, see Ref.~\cite[Sec.~2.5]{Yafaev2022Spectral}. The formulas below are obtained directly for this model, and their residual estimates are proved in Lemma~\ref{lem:approximate-tail-modes}. For $c$ in a compact interval $K$ and $x$ sufficiently large that $|c|<2\lambda x^\alpha$, define
\begin{align}
 \theta(x,c)&\equiv\arccos\frac{c}{2\lambda x^\alpha},&
 \vartheta(x,c)&\equiv\frac{\theta(x,c)}{k},\nonumber\\
 S(x,c)&\equiv\frac{\pi x}{2k}
 +\frac1k\int_x^\infty\arcsin\frac{c}{2\lambda y^\alpha}\,dy,&
 v(x,c)&\equiv2\lambda x^\alpha\sin\theta(x,c).
 \label{eq:tail-phase}
\end{align}
Here $\vartheta$ is the angle conjugate to action, not the physical momentum $p$. The phase is normalized by $S(x,c)-\pi x/(2k)\to0$ at infinity; the integral converges because $\alpha>1$. Also $kv(x,c)$ is the magnitude of the classical action velocity on this energy branch. Differentiation gives
\begin{equation}
 S_x=\vartheta,\qquad S_c=\tau_c(x)\equiv
 \int_x^\infty\frac{dy}{2\lambda k y^\alpha
 \sqrt{1-c^2/(4\lambda^2y^{2\alpha})}},\qquad
 S_{xc}=-\frac1{kv(x,c)}.
 \label{eq:phase-time-identity}
\end{equation}
Thus $\tau_c(x)$ is the classical travel time between action $x$ and infinity, positive and strictly decreasing. At a fixed large action $R$, the source and sink points on the energy level $c$ have angles
\begin{equation}
 \varphi_{l,R}^{\sigma}(c)\equiv\frac{2\pi l+\sigma\theta(R,c)}{k},
 \qquad \sigma\in\{+1,-1\}.
 \label{eq:finite-section-angles}
\end{equation}
The sign $+1$ denotes a source and $-1$ a sink, as in Eq.~\eqref{eq:upper-lower-rays}.

For background on Husimi regularization and semiclassical phase-space measures, see Ref.~\cite{LionsPaul1993}, in particular p.~557, Eq.~(8).

\begin{theorem}[First passage through infinity]
\label{thm:first-passage}
Fix $k>2$, $\lambda>0$, and $U\in U(k)$ in the convention of Eq.~\eqref{eq:WKB-extension-domain}. Let $|\rho_0,\hbar\rangle$ be the normalized coherent state centered at $\rho_0\equiv(q_0,p_0)$, with $c_0\equiv H_k(\rho_0)\ne0$. Suppose its forward classical orbit first escapes at time $t_*$ on the sink ray $j$. Set
\begin{equation}
 \psi_\hbar(t)\equiv e^{-it\widehat H_{k,\hbar}^{U}/\hbar}|\rho_0,\hbar\rangle,
 \qquad d\mu_{u,\hbar}(q,p)\equiv
 \frac{|\langle q,p\mid u\rangle|^2}{2\pi\hbar}\,dq\,dp.
 \label{eq:first-passage-evolution}
\end{equation}
Fix a sufficiently large $R$ with $\tau_{c_0}(R)<t_*$. Define the observation times and phase points by
\begin{equation*}
 t_R^\sigma\equiv t_*+\sigma\tau_{c_0}(R),\qquad
 \rho_{l,R}^{\sigma}
 \equiv\sqrt{2R}\bigl(\cos\varphi_{l,R}^{\sigma}(c_0),
                 \sin\varphi_{l,R}^{\sigma}(c_0)\bigr).
\end{equation*}
Then
\begin{equation}
 \mu_{\psi_\hbar(t_R^-),\hbar}\rightharpoonup\delta_{\rho_{j,R}^{-}},
 \qquad
 \mu_{\psi_\hbar(t_R^+),\hbar}\rightharpoonup
 \sum_{l=0}^{k-1}|(S_{\mathrm q}(U))_{lj}|^2\delta_{\rho_{l,R}^{+}}
 \quad(\hbar\downarrow0).
 \label{eq:first-passage-limits}
\end{equation}
These are weak limits of the full probability measures. More generally, let $\rho_l(s)$ issue from source ray $l$ at $s=0$ with energy $c_0$. For every fixed $0<s<2T(c_0)$,
\begin{equation}
 \mu_{\psi_\hbar(t_*+s),\hbar}\rightharpoonup
 \sum_{l=0}^{k-1}|(S_{\mathrm q}(U))_{lj}|^2\delta_{\rho_l(s)}.
 \label{eq:finite-post-escape-limit}
\end{equation}
Here $T$ is given by Eq.~\eqref{eq:escapeTime}. For each fixed $0\le t<t_*$ the Husimi measure converges to the ordinary classical trajectory. The assertions exclude the escape times themselves.
\end{theorem}

The proof uses finitely many approximate WKB transport corrections with exact boundary values. Square-summability of their residual and full-norm Gaussian matching suffice for Duhamel's formula. In the proof, $K$ is a compact interval with $c_0$ in its interior.

The complex Gaussian phase and its evolution under the linearized Hamiltonian flow are standard ingredients of the complex-germ method; see Ref.~\cite[Sec.~3]{BelovDobrokhotov1992}. Here the required norm approximation is proved directly on the action lattice.

\begin{lemma}[Compact-orbit Gaussian propagation]
\label{lem:compact-action-gaussian}
On a fixed compact time interval on which a nonzero-energy classical orbit remains finite, the exact quantum evolution of a coherent state is approximated in norm, with error $o(1)$, by a sampled Gaussian on the action lattice. The same assertion holds for an initially sampled Gaussian with fixed positive width. It holds for every extension $U$.
\end{lemma}

\begin{proof}
On the full number lattice in Eq.~\eqref{eq:midpoint-action}, the Jacobi expression shifts $n$ by $\pm k$ and has midpoint coefficient $A_\hbar$. Its principal symbol is the action--angle Hamiltonian in Eq.~\eqref{eq:classicalHam}; in this proof write $\mathcal H\equiv H_k$ in those coordinates. Let $(I(t),\varphi(t))$ be the classical orbit in the action--angle variables introduced above, with $\varphi$ lifted continuously to $\mathbb R$. Here $I$ is bounded above and bounded away from zero. Consider
\begin{equation}
 G_{\hbar,n}(t)\equiv\hbar^{1/4}\mathcal N(t)
 \exp\left(\frac{i}{\hbar}
 [\mathcal A(t)+\varphi(t)(x_n-I(t))
                  +Z(t)(x_n-I(t))^2/2]\right),
 \label{eq:sampled-gaussian}
\end{equation}
Here $\mathcal A(t)$ is the classical action accumulated along the orbit and contributes a common scalar phase, whereas the complex parameter $Z(t)$ controls the Gaussian width through $\Im Z(t)>0$ and its quadratic phase through $\Re Z(t)$. The factor $\mathcal N(t)$ is the corresponding normalization, with $|\mathcal N|=(\Im Z/\pi)^{1/4}$. All Hamiltonian derivatives below are evaluated on the orbit:
\begin{align}
 \dot I&=\mathcal H_\vartheta,&
 \dot\varphi&=-\mathcal H_x,&
 \dot{\mathcal A}&=\varphi\dot I-\mathcal H,\nonumber\\
 \dot Z&=-\mathcal H_{xx}-2\mathcal H_{x\vartheta}Z
                         -\mathcal H_{\vartheta\vartheta}Z^2,&
 \frac{\dot{\mathcal N}}{\mathcal N}
 &=-\tfrac12(\mathcal H_{x\vartheta}
                         +\mathcal H_{\vartheta\vartheta}Z).
 \label{eq:action-gaussian-transport}
\end{align}
Let $M_{ij}(t)$ denote the entries of the real symplectic linearized flow matrix. Then
$Z(t)=(M_{21}+M_{22}Z(0))/(M_{11}+M_{12}Z(0))$ and
$\Im Z(t)=\Im Z(0)/|M_{11}+M_{12}Z(0)|^2>0$.
The denominator cannot vanish. Thus the widths stay uniformly positive and bounded on the compact interval, and the normalization condition is preserved.

Here the residual is checked directly in the difference equation. With the local rescaled variable $y\equiv(x-I)/\sqrt\hbar$, the ratio of the shifted Gaussian to the original is
$e^{i\varepsilon k\varphi+i\varepsilon k\sqrt\hbar Zy+ik^2\hbar Z/2}$.
Expansion of the operator through order $\hbar$, divided by the Gaussian, gives
\begin{align*}
 \mathcal H
 &+\sqrt\hbar(\mathcal H_x+\mathcal H_\vartheta Z)y\\
 &+\frac{\hbar}{2}
 (\mathcal H_{xx}+2\mathcal H_{x\vartheta}Z
                       +\mathcal H_{\vartheta\vartheta}Z^2)y^2
 -\frac{i\hbar}{2}
       (\mathcal H_{x\vartheta}+\mathcal H_{\vartheta\vartheta}Z).
\end{align*}
Eq.~\eqref{eq:action-gaussian-transport} cancels these terms in $i\hbar\partial_tG$. Taylor's integral remainder, multiplied by the Gaussian, is bounded by $C\hbar^{7/4}(1+|y|^M)e^{-\beta y^2}$ on a fixed compact action neighborhood, for constants $C,\beta>0$ independent of $\hbar$. The shifted Gaussian factors $e^{C\sqrt\hbar|y|}$ are absorbed in this bound. Multiplication by a fixed smooth action cutoff equal to one near the compact orbit makes the sequence finitely supported and adds $O(\hbar^\infty)$ to the residual. In the rest of the proof, $G_\hbar$ denotes this cut-off sequence, which belongs to the minimal domain of every extension. With
\begin{equation*}
 R_\hbar(t)\equiv\bigl(i\hbar\partial_t-\widehat H_{k,\hbar}^{U}\bigr)G_\hbar(t),
\end{equation*}
summing squares on the $y$ mesh $\sqrt\hbar$ gives, uniformly on the fixed time interval,
\begin{equation*}
 \|R_\hbar(t)\|_{\ell^2}=O(\hbar^{3/2}).
\end{equation*}
If $\psi_\hbar^G(t)$ is the exact solution with $\psi_\hbar^G(0)=G_\hbar(0)$, Duhamel's formula gives, for either sign of $t$ in the fixed compact interval containing zero,
\begin{equation*}
 \|\psi_\hbar^G(t)-G_\hbar(t)\|_{\ell^2}
 \le\frac{|t|}{\hbar}
 \sup_{s\in[\min(0,t),\,\max(0,t)]}\|R_\hbar(s)\|_{\ell^2}
 =O(\sqrt\hbar).
\end{equation*}
Initial matching is only $o_{\ell^2}(1)$, so the lemma's total error is $o(1)$, not necessarily $O(\sqrt\hbar)$.

For the initial coherent state, let $I_0\equiv I(\rho_0)>0$ and $\varphi_0\equiv\varphi(\rho_0)$. Stirling's formula gives
\begin{equation}
 \langle n\mid\rho_0,\hbar\rangle=
 \hbar^{1/4}e^{in\varphi_0}(2\pi I_0)^{-1/4}
 \exp\left(-\frac{(x_n-I_0)^2}{4\hbar I_0}\right)
 +o_{\ell^2}(1).
 \label{eq:coherent-poisson-gaussian}
\end{equation}
It is uniform on each bounded rescaled action interval. To see full-norm convergence, use the exact norm one of the coherent state, the norm tending to one of the sampled normalized Gaussian, and local Riemann summation; then let the rescaled interval grow. This matches Eq.~\eqref{eq:sampled-gaussian} with $Z(0)=i/(2I_0)$, $\mathcal A(0)=I_0\varphi_0$, and $\mathcal N(0)=e^{-i\varphi_0/2}(2\pi I_0)^{-1/4}$.
\end{proof}

\begin{lemma}[Endpoint-normalized approximate modes]
\label{lem:approximate-tail-modes}
Choose $R$ so large that $|c|<2\lambda(R/4)^\alpha$ on a neighborhood of $K$. For every integer $L\ge0$ there are smooth amplitudes $\mathfrak b_{\sigma,\ell}(x,c)$, $0\le\ell\le L$, on $x\ge R/4$, $c\in K$, with
\begin{equation}
 \mathfrak b_{\sigma,0}\equiv v^{-1/2},\qquad
 |\partial_x^a\partial_c^b\mathfrak b_{\sigma,\ell}|
 \le C_{a,b,\ell,R}x^{-\alpha/2-\ell-a}.
 \label{eq:transport-symbol-bounds}
\end{equation}
Define
\begin{equation}
 W_{r,m}^{\sigma,L}(c;\hbar)\equiv
 e^{-i\sigma\pi(r+1/2)/(2k)}
 e^{i\sigma S(x_{r,m},c)/\hbar}
 \sum_{\ell=0}^{L}\hbar^\ell\mathfrak b_{\sigma,\ell}(x_{r,m},c).
 \label{eq:approximate-tail-mode}
\end{equation}
Unlike the exact tail solutions $w_r^{\pm}(z)$ in Theorem~\ref{thm:WKB-boundary-coordinates}, the sequences $W_r^{\sigma,L}(c;\hbar)$ are finite-order semiclassical parametrices: they solve the tail recurrence only up to the residual estimated below, while providing estimates uniform in the semiclassical regime.
For sufficiently small $\hbar$ and $x_{r,m}\ge R/3$, its tail residual satisfies
\begin{equation}
 |((J_{r,\hbar}-c)W^{\sigma,L})_m|
 \le C_{L,R}\hbar^{L+2}x_{r,m}^{\alpha/2-L-2}.
 \label{eq:tail-residual-pointwise}
\end{equation}
If $L>(\alpha-3)/2$, this residual is in $\ell^2$ on $x_{r,m}\ge R/3$, with norm $O_R(\hbar^{L+3/2})$. For such $L$, choose a smooth cutoff $0\le\omega_R\le1$ vanishing for $x\le R/2$ and equal to one for $x\ge3R/4$. Then $\omega_RW^{\sigma,L}$ belongs to the maximal domain, and
\begin{equation}
 \Gamma_{\sigma,\hbar}(\omega_RW_r^{\sigma,L})=e_r,\qquad
 \Gamma_{-\sigma,\hbar}(\omega_RW_r^{\sigma,L})=0.
 \label{eq:approximate-mode-boundary}
\end{equation}
The sign subscript on $\Gamma$ means $+$ or $-$.
\end{lemma}

\begin{proof}
Let $\mathcal S_{\rm tail}^{\mu}$ consist of functions satisfying $|\partial_x^a\partial_c^b g|\le C_{a,b}x^{\mu-a}$ on this tail. The coefficient $A_\hbar$ is $x^\alpha$ times a smooth function of $\hbar/x$, and the arguments of the inverse trigonometric functions stay inside a fixed compact subinterval of $(-1,1)$. All their indicated derivatives have the corresponding symbol bounds.

For a smooth amplitude $g$, the conjugated recurrence is exactly
\begin{equation}
 \mathcal T_{\sigma,\hbar}g\equiv
 \sum_{\varepsilon=\pm1}A_\hbar(x+\varepsilon k\hbar/2)
 \exp\left(i\sigma\varepsilon k\int_0^1\vartheta(x+s\varepsilon k\hbar,c)\,ds\right)
 g(x+\varepsilon k\hbar)-cg(x).
 \label{eq:conjugated-recurrence}
\end{equation}
Denote its Taylor coefficient at order $\hbar^d$ by $\mathcal T_{\sigma}^{[d]}$; the bracketed index distinguishes this coefficient from the full operator $\mathcal T_{\sigma,\hbar}$. Taylor's integral remainder gives
\begin{equation}
 \mathcal T_{\sigma,\hbar}g=
 \sum_{d=1}^{M}\hbar^d\mathcal T_{\sigma}^{[d]}g+
 \hbar^{M+1}\mathcal R_{\sigma,M,\hbar}g,\qquad
 \mathcal T_{\sigma}^{[d]}:\mathcal S_{\rm tail}^{\mu}
       \longrightarrow\mathcal S_{\rm tail}^{\mu+\alpha-d}.
 \label{eq:transport-expansion}
\end{equation}
The remainder has the same mapping property with $d=M+1$, uniformly for small $\hbar$, including any fixed $x,c$ derivatives. Indeed, each $\hbar$ derivative of the coefficient or exponential in Eq.~\eqref{eq:conjugated-recurrence} lowers its $x$ order by one. A derivative of the shifted amplitude does the same, and all shifted arguments are comparable to $x$. This proves the remainder assertion directly. The zeroth coefficient vanishes by $2\lambda x^\alpha\cos(k\vartheta)=c$. The first coefficient is
\begin{equation}
 \mathcal T_{\sigma}^{[1]}g=i\sigma k(vg_x+\tfrac12v_xg).
 \label{eq:first-transport-operator}
\end{equation}
Starting with the leading amplitude in Eq.~\eqref{eq:transport-symbol-bounds}, recursively define, for $\ell\ge1$,
\begin{align}
 \mathcal E_{\sigma,\ell}
 &\equiv\sum_{d=2}^{\ell+1}\mathcal T_{\sigma}^{[d]}
                 \mathfrak b_{\sigma,\ell+1-d},\nonumber\\
 \mathfrak b_{\sigma,\ell}(x,c)
 &\equiv\frac{\mathfrak b_{\sigma,0}(x,c)}{i\sigma k}
   \int_x^\infty
   \frac{\mathcal E_{\sigma,\ell}(y,c)}
        {v(y,c)\mathfrak b_{\sigma,0}(y,c)}\,dy.
 \label{eq:transport-recursion}
\end{align}
Inductively $\mathcal E_{\sigma,\ell}\in\mathcal S_{\rm tail}^{\alpha/2-\ell-1}$, so the integrand is in $\mathcal S_{\rm tail}^{-\ell-1}$. It is integrable with all energy derivatives. This proves Eq.~\eqref{eq:transport-symbol-bounds}. Differentiating the integral gives $\mathcal T_{\sigma}^{[1]}\mathfrak b_{\sigma,\ell}=-\mathcal E_{\sigma,\ell}$. Thus the terms of orders $1,\ldots,L+1$ cancel, and the Taylor remainder proves Eq.~\eqref{eq:tail-residual-pointwise}. Squared summation on the mesh $k\hbar$ gives
\begin{equation*}
 \sum_{x_{r,m}\ge R/3}|((J_{r,\hbar}-c)W^{\sigma,L})_m|^2
 \le C_R\hbar^{2L+3}\int_{R/4}^\infty x^{\alpha-2L-4}\,dx.
\end{equation*}
The integral converges under the stated condition on $L$. Also $W^{\sigma,L}\in\ell^2$ since $\alpha>1$. After multiplication by $\omega_R$, the extra commutator has finite support at each fixed $\hbar$, proving maximal-domain membership. At fixed $\hbar$ and $m\to\infty$,
\begin{equation*}
 W_{r,m}^{\sigma,L}=
 \frac{e^{i\sigma\pi m/2}}{\sqrt{2b_{m,\hbar}^{(r)}}}(1+o(1)).
\end{equation*}
Thus, with the normalization of the exact modes in Eq.~\eqref{eq:WKB-solutions},
\begin{equation*}
 W_{r,m}^{\sigma,L}(c;\hbar)
 =2^{-1/2}w_{r,m}^{\sigma}(c)(1+o(1)),
 \qquad m\to\infty.
\end{equation*}
This identifies their leading tail asymptotics and is not an exact equality of solutions.
Indeed $S-\pi x/(2k)\to0$, each correction ratio tends to zero, and the residue phase cancels the half-unit action shift exactly. Eq.~\eqref{eq:boundary-maps} proves Eq.~\eqref{eq:approximate-mode-boundary}.
\end{proof}

\begin{lemma}[Propagation of a tail packet]
\label{lem:tail-packet-propagation}
Let $f$ be a normalized centered complex Gaussian on $\mathbb R$, let $\zeta\in C_c^\infty(K)$ equal one near $c_0$, and put
\begin{equation}
 f_\hbar(c)\equiv\hbar^{-1/4}\zeta(c)
                 f\left(\frac{c-c_0}{\sqrt\hbar}\right).
 \label{eq:gaussian-energy-profile}
\end{equation}
In the proof of Theorem~\ref{thm:first-passage}, Lemma~\ref{lem:leading-gaussian-matching} selects the complex width of $f$ from the incoming sampled Gaussian and the derivatives of the eikonal $S$. Thus $f_\hbar$ is an auxiliary profile constructed from the prescribed coherent initial state, not an additional datum in the statement of the theorem.
For $\mathbf v\in\mathbb C^k$ and $L>(\alpha-3)/2$, define
\begin{equation}
 \Psi_\hbar(t)_{r,m}\equiv\frac{\omega_R(x_{r,m})}{\sqrt{2\pi}}
 \int_K e^{-ic(t-t_*)/\hbar}f_\hbar(c)
 \left[(U\mathbf v)_rW_{r,m}^{+,L}(c;\hbar)
               +v_rW_{r,m}^{-,L}(c;\hbar)\right]dc.
 \label{eq:approximate-tail-packet}
\end{equation}
Then $\Psi_\hbar(t)\in D(\widehat H_{k,\hbar}^U)$, is differentiable in Hilbert norm, and
\begin{equation}
 \sup_{t_R^-\le t\le t_R^+}
 \left\|e^{-i(t-t_R^-)\widehat H_{k,\hbar}^U/\hbar}
       \Psi_\hbar(t_R^-)-\Psi_\hbar(t)\right\|
 =O_R(\hbar^{L+3/4})+O_R(\hbar^\infty).
 \label{eq:tail-packet-duhamel}
\end{equation}
The estimates are uniform for $\mathbf v$ in a bounded set.
\end{lemma}

\begin{proof}
At each fixed $\hbar$, the integrand is continuous in the maximal graph norm, by Lemma~\ref{lem:approximate-tail-modes}. The boundary maps are continuous in that norm, so Eq.~\eqref{eq:approximate-mode-boundary} passes through the integral and gives $\Gamma_{+,\hbar}\Psi_\hbar=U\Gamma_{-,\hbar}\Psi_\hbar$. Multiplication by $c$ under the compact energy integral justifies time differentiation. The tail residual is bounded uniformly in time by
\begin{equation}
 \|{\rm residual}_{\rm tail}\|
 \le C_R\hbar^{L+3/2}\|f_\hbar\|_{L^1(K)}
 \le C_R'\hbar^{L+7/4}.
 \label{eq:packet-tail-residual}
\end{equation}
The following Fourier estimate controls the commutator with $\omega_R$ and the global tails needed for matching. Put $s\equiv t-t_*$ and $\eta\equiv(c-c_0)/\sqrt\hbar$, and extract the linear energy oscillation $\exp(i\eta[\sigma\tau_{c_0}(x)-s]/\sqrt\hbar)$. The remaining energy phase is
\begin{equation*}
 \Phi_{\sigma,\hbar}(x,\eta)\equiv\frac{\sigma}{\hbar}
 [S(x,c_0+\sqrt\hbar\eta)-S(x,c_0)
                   -\sqrt\hbar\eta S_c(x,c_0)].
\end{equation*}
Taylor's integral formula gives $|\partial_\eta \Phi_{\sigma,\hbar}|\le C_R|\eta|$; all higher derivatives are polynomially bounded uniformly for $x\ge R/3$. After multiplication by $f$, the cutoff and the transport amplitude, any fixed number of $\eta$ derivatives has $L^1(d\eta)$ norm at most $C_{R,M}x^{-\alpha/2}$. Fourier integration by parts gives
\begin{equation}
 |\Psi_\hbar^\sigma(t)_{r,m}|
 \le C_{R,M}\hbar^{1/4}x_{r,m}^{-\alpha/2}
 \left(1+
 \frac{|\sigma\tau_{c_0}(x_{r,m})-s|}{\sqrt\hbar}\right)^{-M}.
 \label{eq:packet-global-decay}
\end{equation}
Here $\Psi^\sigma$ denotes the corresponding uncut branch, with its channel coefficient included in $C_{R,M}$.

On the commutator support, the action is between $R/2+O(\hbar)$ and $3R/4+O(\hbar)$, hence $\tau_{c_0}(x)>\tau_{c_0}(R)$ by a fixed amount. For $|s|\le\tau_{c_0}(R)$ both branches and their lattice shifts there are bounded by arbitrary powers of $\hbar$. There are $O_R(\hbar^{-1})$ sites with bounded recurrence coefficients, so the commutator residual is $O_R(\hbar^\infty)$ in norm. Duhamel and division of Eq.~\eqref{eq:packet-tail-residual} by $\hbar$ give Eq.~\eqref{eq:tail-packet-duhamel}.
\end{proof}

\begin{lemma}[Gaussian matching on the action lattice]
\label{lem:leading-gaussian-matching}
At $t_R^\sigma$, the matching branch of Eq.~\eqref{eq:approximate-tail-packet} has, up to a scalar phase independent of $r$, a full-norm approximation
\begin{equation}
 u_n=\hbar^{1/4}\sum_{l=0}^{k-1}d_l
 e^{in\varphi_{l,R}^{\sigma}(c_0)}
 g_\sigma\left(\frac{\hbar(n+1/2)-R}{\sqrt\hbar}\right)+o_{\ell^2}(1),
 \label{eq:lattice-gaussian-branches}
\end{equation}
where $\|g_\sigma\|_2=1$ and
\begin{equation}
 \mathbf d\equiv B_\sigma\mathbf v_\sigma,\qquad
 B_\sigma\equiv e^{-i\sigma\pi/(4k)}F_{-\sigma},\qquad
 \mathbf v_-\equiv\mathbf v,\quad \mathbf v_+\equiv U\mathbf v.
 \label{eq:finite-section-channel-matrix}
\end{equation}
The opposite branch is $o_{\ell^2}(1)$. Every normalized sampled Gaussian centered at $(R,\varphi_{j,R}^-(c_0))$ can be matched in full norm to $o(1)$ by choosing one Gaussian $f$ and $\mathbf v=B_-^*e_j$. Its outgoing coefficients are
\begin{equation}
 \mathbf d^{\rm out}\equiv B_+UB_-^*e_j
          =e^{-i\pi/(2k)}S_{\mathrm q}(U)e_j.
 \label{eq:matched-source-coefficients}
\end{equation}
\end{lemma}

\begin{proof}
Use the energy variable $\eta$ from the preceding proof and put $y\equiv(x-R)/\sqrt\hbar$. At $t_R^\sigma$ the linear term in $\eta$ cancels. Define the real scalars
\begin{equation*}
 S_{xx}^{\sigma}\equiv\sigma S_{xx}(R,c_0),\qquad
 S_{xc}^{\sigma}\equiv\sigma S_{xc}(R,c_0)\ne0,\qquad
 S_{cc}^{\sigma}\equiv\sigma S_{cc}(R,c_0).
\end{equation*}
The remaining quadratic phase is
$S_{xx}^{\sigma} y^2/2+S_{xc}^{\sigma} y\eta+S_{cc}^{\sigma}\eta^2/2$.
Since $\mathfrak b_{\sigma,0}(R,c_0)/\sqrt{k}=\sqrt{|S_{xc}^{\sigma}|}$, the normalized leading transform is
\begin{equation}
 g_\sigma(y)\equiv\sqrt{\frac{|S_{xc}^{\sigma}|}{2\pi}}\,
 e^{iS_{xx}^{\sigma} y^2/2}
 \int_{\mathbb R}e^{iS_{xc}^{\sigma} y\eta+iS_{cc}^{\sigma}\eta^2/2}
 f(\eta)\,d\eta.
 \label{eq:quadratic-matching-transform}
\end{equation}
Plancherel's identity makes this transform unitary: its factors are quadratic-phase multipliers, a Fourier transform and a dilation. To make the inverse entirely explicit, for a desired normalized $g(y)\equiv\gamma\exp(-\kappa y^2/2)$, $\Re\kappa>0$, put
\begin{equation}
 \kappa_E\equiv iS_{cc}^{\sigma}+
       \frac{(S_{xc}^{\sigma})^2}{\kappa+iS_{xx}^{\sigma}},
 \qquad
 f(\eta)\equiv\frac{\gamma\sqrt{\kappa_E-iS_{cc}^{\sigma}}}
                  {\sqrt{|S_{xc}^{\sigma}|}}\exp(-\kappa_E\eta^2/2).
 \label{eq:explicit-energy-gaussian}
\end{equation}
Here $\Re\kappa_E=(S_{xc}^{\sigma})^2\Re\kappa/|\kappa+iS_{xx}^{\sigma}|^2>0$, and the square root is chosen with positive real part. The elementary Gaussian integral gives Eq.~\eqref{eq:quadratic-matching-transform} with output $g$ and $\|f\|_2=1$.

Taylor's formula in the energy integral proves the local asymptotic uniformly for bounded $y$, by dominated convergence with the Gaussian. To obtain full norm, use Eq.~\eqref{eq:packet-global-decay}. Near $R$, the quantity $|\tau_{c_0}(x)-\tau_{c_0}(R)|$ is comparable to $|x-R|$, so the rescaled envelopes are bounded by $C_M(1+|y|)^{-M}$. Outside a fixed neighborhood of $R$, and on the opposite branch everywhere, the energy phase derivative is separated from zero. Squared summation uses $\sum_{x_{r,m}\ge R/3}x_{r,m}^{-\alpha}=O_R(\hbar^{-1})$ and gives $O_R(\hbar^\infty)$ there. Dominated Riemann summation then proves full-norm convergence.

To track residues, the common action phase leaves the vector
$e^{-i\sigma\pi(r+1/2)/(2k)}(\mathbf v_\sigma)_r$.
Apply the Fourier matrix with entries $k^{-1/2}e^{-2\pi ilr/k}$ to this vector. The result is $e^{-i\sigma\pi/(4k)}F_{-\sigma}\mathbf v_\sigma$. Inverting the transform introduces factors $e^{2\pi ilr/k}$. These equal $e^{2\pi iln/k}$ on the sector $n=r+km$. Combining these with the linear action phase proves Eqs.~\eqref{eq:lattice-gaussian-branches} and \eqref{eq:finite-section-channel-matrix}. Choosing $\mathbf v=B_-^*e_j$ isolates the incoming ray, and Eq.~\eqref{eq:explicit-energy-gaussian} matches the Gaussian produced by Lemma~\ref{lem:compact-action-gaussian}. Any common scalar phase, possibly depending on $\hbar$, is absorbed in $f_\hbar$. This also proves Eq.~\eqref{eq:matched-source-coefficients}.
\end{proof}

\begin{lemma}[Husimi concentration of separated packets]
\label{lem:sampled-gaussian-husimi}
A normalized sampled action Gaussian
$\hbar^{1/4}e^{in\varphi}g((x_n-R)/\sqrt\hbar)$
has Husimi measure converging to the point of action $R$ and angle $\varphi$. For a finite sum of such packets at distinct points, the limiting measure is the sum of their point masses with weights equal to the squared moduli of their coefficients.
\end{lemma}

\begin{proof}
Write $w\equiv(q+ip)/\sqrt2$ and $w_0\equiv\sqrt R e^{i\varphi}$. Shifting $n$ by one in the sampled Gaussian and using its smooth rapidly decreasing envelope gives
$\|(\sqrt\hbar\,a_\hbar-w_0)u_\hbar\|\to0$.
The coherent-state resolution of the identity, also verified by integrating its number-basis expansion, gives
\begin{equation}
 \int|w-w_0|^2\,d\mu_{u_\hbar,\hbar}
 =\|(\sqrt\hbar\,a_\hbar-w_0)u_\hbar\|^2
                              +\hbar\|u_\hbar\|^2.
 \label{eq:husimi-second-moment}
\end{equation}
This proves concentration. For a bounded test function $b$, the same resolution and Cauchy--Schwarz give
\begin{equation}
 \left|\int b\,d\mu_{u,\hbar}-\int b\,d\mu_{v,\hbar}\right|
 \le\|b\|_\infty(\|u\|+\|v\|)\|u-v\|.
 \label{eq:husimi-norm-stability}
\end{equation}
Thus full-norm $o(1)$ approximations preserve the weak limit. For two distinct packet centers choose disjoint fixed neighborhoods. On either neighborhood the coherent-state transform of the other packet has norm $o(1)$; outside both neighborhoods both transforms have norm $o(1)$. Cauchy--Schwarz proves that every mixed Husimi integral tends to zero. This proves the statement for a finite sum.
\end{proof}

\begin{proof}[Proof of Theorem~\ref{thm:first-passage}]
For fixed $R$, the classical orbit from $\rho_0$ to $\rho_{j,R}^-$ remains in a compact action interval in $(0,\infty)$. Lemma~\ref{lem:compact-action-gaussian} approximates $\psi_\hbar(t_R^-)$ in full norm by a sampled Gaussian. Lemma~\ref{lem:leading-gaussian-matching} chooses $f$ and $\mathbf v=B_-^*e_j$ to match it by $\Psi_\hbar(t_R^-)$ to $o(1)$. Take, for example, $L=k$ in Lemma~\ref{lem:tail-packet-propagation}. Its error tends to zero. Unitarity gives
\begin{equation*}
 \|\psi_\hbar(t_R^+)-\Psi_\hbar(t_R^+)\|=o(1).
\end{equation*}
The outgoing full-norm decomposition is Eq.~\eqref{eq:lattice-gaussian-branches} with coefficients in Eq.~\eqref{eq:matched-source-coefficients}. Their common extra factor has modulus one. Lemma~\ref{lem:sampled-gaussian-husimi} and Eq.~\eqref{eq:husimi-norm-stability} prove Eq.~\eqref{eq:first-passage-limits}.

The outgoing envelopes are again Gaussians. For fixed $0<s<2T(c_0)$ choose $R$ such that $\tau_{c_0}(R)<s$. Each orbit segment from $\rho_{l,R}^+$ to $\rho_l(s)$ has action in a compact subset of $(0,\infty)$, including its finite turning point if present. Lemma~\ref{lem:compact-action-gaussian} propagates each component and Lemma~\ref{lem:sampled-gaussian-husimi} gives Eq.~\eqref{eq:finite-post-escape-limit}. The pre-escape assertion follows from the same compact-orbit construction.
\end{proof}

\begin{corollary}[Angular shell probabilities]
\label{cor:angular-shell-probabilities}
Let $0\le\omega_0\le1$ be smooth, supported in $(-2,2)$ and equal to one near zero. Normalize the angular Husimi measure by its mass in the shell weighted by $\omega_0(I-R)$. At both $t_R^\pm$ this denominator tends to one, hence is at least $1/2$ for sufficiently small $\hbar$. The conditional angular measures converge respectively to
\begin{equation*}
 \delta_{\varphi_{j,R}^-(c_0)},\qquad
 \sum_l|(S_{\mathrm q}(U))_{lj}|^2\delta_{\varphi_{l,R}^+(c_0)}.
\end{equation*}
Replacing finite-section angles by $\varphi_l^\pm$ incurs an $O_{c_0}(R^{-\alpha})$ error in bounded-Lipschitz distance. Thus the iterated limit $\hbar\downarrow0$, then $R\to\infty$, gives the asymptotic-ray probabilities $|(S_{\mathrm q}(U))_{lj}|^2$.
\end{corollary}

\begin{proof}
Apply Eq.~\eqref{eq:first-passage-limits} to the shell weight and to its product with a continuous angular test function. The limit has mass one on $I=R$. The angular estimate follows from Eq.~\eqref{eq:finite-section-angles} and $\arccos z=\pi/2+O(z)$ at zero.
\end{proof}

\begin{corollary}[Coherent returns and the scope of stochastic gluing]
\label{cor:coherent-returns}
In the setting of Theorem~\ref{thm:first-passage}, assume $c_0>0$. For every fixed integer $\nu\ge1$ and $0<s<2T(c_0)$, the full Husimi measures satisfy
\begin{equation}
 \mu_{\psi_\hbar(t_*+2(\nu-1)T(c_0)+s),\hbar}
 \rightharpoonup
 \sum_{l=0}^{k-1}\left|\bigl(S_{\mathrm q}(U)^\nu e_j\bigr)_l\right|^2
 \delta_{\rho_l(s)}
 \quad(\hbar\downarrow0).
 \label{eq:coherent-return-weights}
\end{equation}
Here $\rho_l(s)$ is the classical branch in Theorem~\ref{thm:first-passage}. In general the limiting source-channel weights differ from $\mathsf P(U)^\nu e_j$, with $\mathsf P(U)$ the one-pass probability matrix defined in Sec.~\ref{sec:WKB-correspondence}.
\end{corollary}

\begin{proof}
For positive energy, source ray $l$ connects through its finite sector to sink ray $l$. The $k$ sectors are rotations by $2\pi/k$ of each other. In the sampled action representation this rotation is multiplication by $e^{2\pi iln/k}$. The bulk Jacobi expression commutes with it, and the Gaussian construction in Lemma~\ref{lem:compact-action-gaussian} is covariant under this multiplication. The propagation phases are therefore common in these packet frames. All branches return after the same time $2T(c_0)$ with one common Gaussian profile and their relative channel amplitudes preserved. Lemmas~\ref{lem:tail-packet-propagation} and \ref{lem:leading-gaussian-matching} accept any input vector with that common profile and multiply it by $S_{\mathrm q}(U)$ up to a common scalar phase. Induction gives the finite Gaussian decomposition after any fixed number of passages, and Lemma~\ref{lem:sampled-gaussian-husimi} proves Eq.~\eqref{eq:coherent-return-weights}.

For example, choose $S_{\mathrm q}(U)$ to be the normalized discrete Fourier matrix; it is realized by $U=F_-^*S_{\mathrm q}(U)F_+$. Its first column has equal squared moduli, whereas its square is a permutation. After two passages the limiting quantum measure has one outgoing branch; the twice-iterated uniform Markov matrix still assigns equal weights. Thus repeated Markov evolution requires an incoherent input at each passage. Spatial separation alone does not remove the relative phases.
\end{proof}

For Husimi limits and wave-packet revivals under free motion on a torus, see Ref.~\cite{Trushechkin2017}. The returns in Corollary~\ref{cor:coherent-returns} arise from passages through infinity, with amplitudes determined by the extension.

\begin{proposition}[Liouville amplitude representation of the Husimi limit]
\label{prop:liouville-amplitude-limit}
Let $U$ and $\rho_0$ satisfy the hypotheses of Theorem~\ref{thm:first-passage}, and put $S_{\mathrm q}\equiv S_{\mathrm q}(U)$. In the source-based component enumeration of Sec.~\ref{sec:classical-flow}, define
\begin{equation}
 Q_c\equiv
 \begin{cases}
 \mathbf 1_k,&c>0,\\
 P_{+1},&c<0,
 \end{cases}
 \qquad
 (P_{+1}z)_l\equiv z_{l-1},\qquad
 J_U(c)\equiv Q_c^*S_{\mathrm q}^*,
 \label{eq:liouville-quantum-boundary}
\end{equation}
where indices are taken modulo $k$. Let $\widetilde L_{J_U}$ be the energy-fiber preserving skew-self-adjoint extension in Proposition~\ref{prop:energy-fiber-extensions}, with domain condition $\tau_+\chi=J_U(c)\tau_-\chi$, and set
\begin{equation}
 V_U(t)\equiv e^{-t\widetilde L_{J_U}}.
 \label{eq:liouville-amplitude-group}
\end{equation}
This group preserves the norm and the energy fibers. Its forward transport and boundary scattering are
\begin{equation}
 \partial_t\chi+\partial_s\chi=0,\qquad
 \chi^{\mathrm{snk}}\equiv Q_c\tau_+\chi,\qquad
 \chi^{\mathrm{src}}\equiv\tau_-\chi=S_{\mathrm q}\chi^{\mathrm{snk}}.
 \label{eq:liouville-amplitude-scattering}
\end{equation}

Write $\rho_0$ as $(\ell_0,c_0,s_0)$ in these coordinates, with $\ell_0$ the source label of its component. Choose any $g\in C_c^\infty(\mathbb R^2)$ with $\|g\|_{L^2}=1$. For sufficiently small $\varepsilon>0$, the normalized initial amplitude
\begin{equation}
 \chi_\varepsilon(c,s)\equiv\varepsilon^{-1}
 g\left(\frac{c-c_0}{\varepsilon},
        \frac{s-s_0}{\varepsilon}\right)e_{\ell_0}
 \label{eq:localized-liouville-amplitude}
\end{equation}
is supported inside the initial component. Identify $V_U(t)\chi_\varepsilon$ with a function on the original phase plane by Proposition~\ref{prop:energy-fiber-extensions}, and define the probability measure
\begin{equation*}
 d\mu^{\mathrm L}_{\varepsilon,t}
 \equiv|V_U(t)\chi_\varepsilon|^2\,dq\,dp.
\end{equation*}
For every $b\in C_b(\mathbb R^2)$ and every fixed $0\le t<t_*+2T(c_0)$ with $t\ne t_*$,
\begin{equation}
 \lim_{\varepsilon\downarrow0}\int b\,d\mu^{\mathrm L}_{\varepsilon,t}
 =
 \lim_{\hbar\downarrow0}\int b\,d\mu_{\psi_\hbar(t),\hbar}.
 \label{eq:liouville-husimi-limit}
\end{equation}
If $c_0>0$, Eq.~\eqref{eq:liouville-husimi-limit} holds at every fixed $t\ge0$ except the escape times $t_*+2mT(c_0)$, $m=0,1,\ldots$. In particular, at the times in Corollary~\ref{cor:coherent-returns}, both limits have branch weights $|(S_{\mathrm q}^{\nu}e_j)_l|^2$.
\end{proposition}

\begin{proof}
A component with index $l$ begins at source ray $l$. Its sink ray is $l$ for positive energy and $l+1$ for negative energy, by the sector geometry in Sec.~\ref{sec:WKB-correspondence}. Hence $Q_c$ converts the component endpoint vector $\tau_+\chi$ to the geometrically labeled sink amplitudes. Eq.~\eqref{eq:liouville-quantum-boundary} gives a measurable unitary family. Proposition~\ref{prop:energy-fiber-extensions} therefore supplies the claimed skew-self-adjoint operator and its unitary group, in the boundary-extension framework of Ref.~\cite[Theorem~19]{Sakbaev2025}. Commutation with multiplication by $c$ follows from its fiberwise construction.

The sign in Eq.~\eqref{eq:liouville-amplitude-group} makes the physical transport advance toward increasing $s$: before meeting an endpoint it is $\chi(c,s,t)=\chi(c,s-t,0)$. Inverting the domain condition gives
\begin{equation*}
 \tau_-\chi=J_U(c)^*\tau_+\chi
          =S_{\mathrm q}Q_c\tau_+\chi,
\end{equation*}
which proves Eq.~\eqref{eq:liouville-amplitude-scattering}. Thus the forward component-to-component scattering matrix is $S_{\mathrm q}Q_c$, while its sink-to-source expression in geometric ray labels is $S_{\mathrm q}$.

The initial support in Eq.~\eqref{eq:localized-liouville-amplitude} stays away from $c=0$ and from the component endpoints when $\varepsilon$ is small. Its norm is one by the change of variables $c=c_0+\varepsilon u$, $s=s_0+\varepsilon v$ and the measure identity $dq\,dp=dc\,ds$. The central first escape time is $t_*=T(c_0)-s_0$. At any observation time allowed in the statement, the central trajectory is separated from its escape times. Continuity of $T(c)$ implies that every characteristic issued from the shrinking support has the same number $\nu$ of endpoint passages as the central one, for all sufficiently small $\varepsilon$. Here $\nu=0$ before the first escape and $\nu=1$ in the stated post-escape interval; for $c_0>0$ any fixed $\nu$ is allowed.

Solving the transport equation along characteristics and applying the boundary condition at each passage gives, on the propagated support,
\begin{equation}
 (V_U(t)\chi_\varepsilon)_l(c,s)
 =
 \bigl[(S_{\mathrm q}Q_{c_0})^{\nu}e_{\ell_0}\bigr]_l
 \varepsilon^{-1}
 g\left(\frac{c-c_0}{\varepsilon},
        \frac{s-t+2\nu T(c)-s_0}{\varepsilon}\right).
 \label{eq:liouville-characteristic-packet}
\end{equation}
The amplitude vanishes outside that support. All components have the same interval length $2T(c)$, so the scalar profile in Eq.~\eqref{eq:liouville-characteristic-packet} is common to all branches. The translation $(c,s)\mapsto(c,s-t+2\nu T(c))$ preserves $dc\,ds$, and the coefficient vector has norm one.

For $\nu=0$ the measure therefore concentrates at the ordinary classical position at time $t$. For $\nu=1$, the initial component terminates on sink ray $j$, so $Q_{c_0}e_{\ell_0}=e_j$. Its outgoing vector is $S_{\mathrm q}e_j$, and each scalar profile concentrates at the corresponding point $\rho_l(t-t_*)$. For positive energy $Q_{c_0}=\mathbf 1_k$ and $\ell_0=j$, so after $\nu$ passages the vector is $S_{\mathrm q}^{\nu}e_j$. The centers are exactly those in Corollary~\ref{cor:coherent-returns}. Weak convergence follows by changing variables in each compactly supported profile and applying continuity of $b$ at its limiting center. Theorem~\ref{thm:first-passage} and Corollary~\ref{cor:coherent-returns} give the same limits for the quantum Husimi measures, proving Eq.~\eqref{eq:liouville-husimi-limit}.
\end{proof}

\section{Numerical illustration}
\label{sec:numerical-illustration}

We take $k=3$, $\lambda=1$, $\hbar=0.0025$, and an initial coherent state with $I_0=1$, $\varphi_0=-\pi/12$, and positive energy $c_0=\sqrt2$. Its center is $q_0=(\sqrt3+1)/2$, $p_0=(1-\sqrt3)/2$, from the action--angle coordinates in Sec.~\ref{sec:classical-flow}. The coherent-state overlap gives its initial Husimi density
\begin{equation}
 Q_{\hbar,0}(q,p)\equiv\frac1{2\pi\hbar}
 \exp\left(-\frac{(q-q_0)^2+(p-p_0)^2}{2\hbar}\right).
 \label{eq:numerical-initial-density}
\end{equation}
We compare the quantum Husimi distribution with a classical ensemble sampled from this Gaussian, rather than a single trajectory; see Ref.~\cite{TrushechkinVolovich2009} for the distribution-based viewpoint. The ensemble follows the classical trajectories, retaining channel amplitudes as in Proposition~\ref{prop:liouville-amplitude-limit}. A separate comparison uses the Markov rule $\mathsf P(U)$, which discards relative phases.

We use Fock--Galerkin compressions with even dimension in every residue sector. Their unitary groups converge strongly at fixed $\hbar$~\cite[Proposition~IV.1 and Theorem~IV.2]{Ashhab2026}, selecting $U=-\mathbf 1_k$ in the boundary coordinates of Eq.~\eqref{eq:boundary-maps}; odd sector dimensions select $U=+\mathbf 1_k$ (Proposition~\ref{prop:quantized-deterministic-gluings}). The figures use $N=51200$ states per sector, cutoff action $3\hbar N=384$, and $U=-\mathbf 1_3$.

The extension $U=-\mathbf1_3$ selected by these even truncations does not give a deterministic ray gluing in the sense of Proposition~\ref{prop:quantized-deterministic-gluings}. Its matrix $S_{\mathrm q}(U)$ is not monomial: at one passage, the amplitude of each sink ray is distributed among all three source rays, as Eq.~\eqref{eq:cubic-channel-cycle} below shows. This is unitary mixing of complex amplitudes that preserves interchannel coherence.

Substitution into Eq.~\eqref{eq:channel-boundary-condition}, using the Fourier factorization in the proof of Proposition~\ref{prop:quantized-deterministic-gluings}, gives
\begin{equation}
 \begin{gathered}
 S_{\mathrm q}(U)=-F_0^*
 \operatorname{diag}(e^{-i\pi r/3})_{r=0}^{2}F_0,\\
 S_{\mathrm q}(U)^2=P_{+1}^*,\qquad
 S_{\mathrm q}(U)^6=\mathbf 1_3.
 \end{gathered}
 \label{eq:cubic-channel-cycle}
\end{equation}
Eq.~\eqref{eq:cubic-channel-cycle} follows by squaring the diagonal phases, which gives the inverse of the shift in Eq.~\eqref{eq:liouville-quantum-boundary}. By Corollary~\ref{cor:coherent-returns}, after $2m$ passages the limiting packet occupies sector $-m\pmod3$ alone; after an odd number it occupies all three, with cyclic permutations of the weights $(4/9,1/9,4/9)$. This pattern repeats after six passages, or time $12T(c_0)$ at a fixed offset from escape. Exact periodicity of the finite-$\hbar$ truncated evolution is not asserted. The choice $U=+\mathbf 1_3$ gives the same limiting probabilities, since its ray matrix differs only by an overall sign.

Fig.~\ref{fig:husimi-liouville} shows the distributions at identical physical times $t\approx0.549$ and $t\approx1.457$, after the first and second escapes of the central trajectory, without fitted alignment. The first passage populates three sectors, whereas the second recombines the limiting packet into sector $2$.

\begin{figure}[H]
 \centering
 \includegraphics[width=\textwidth]{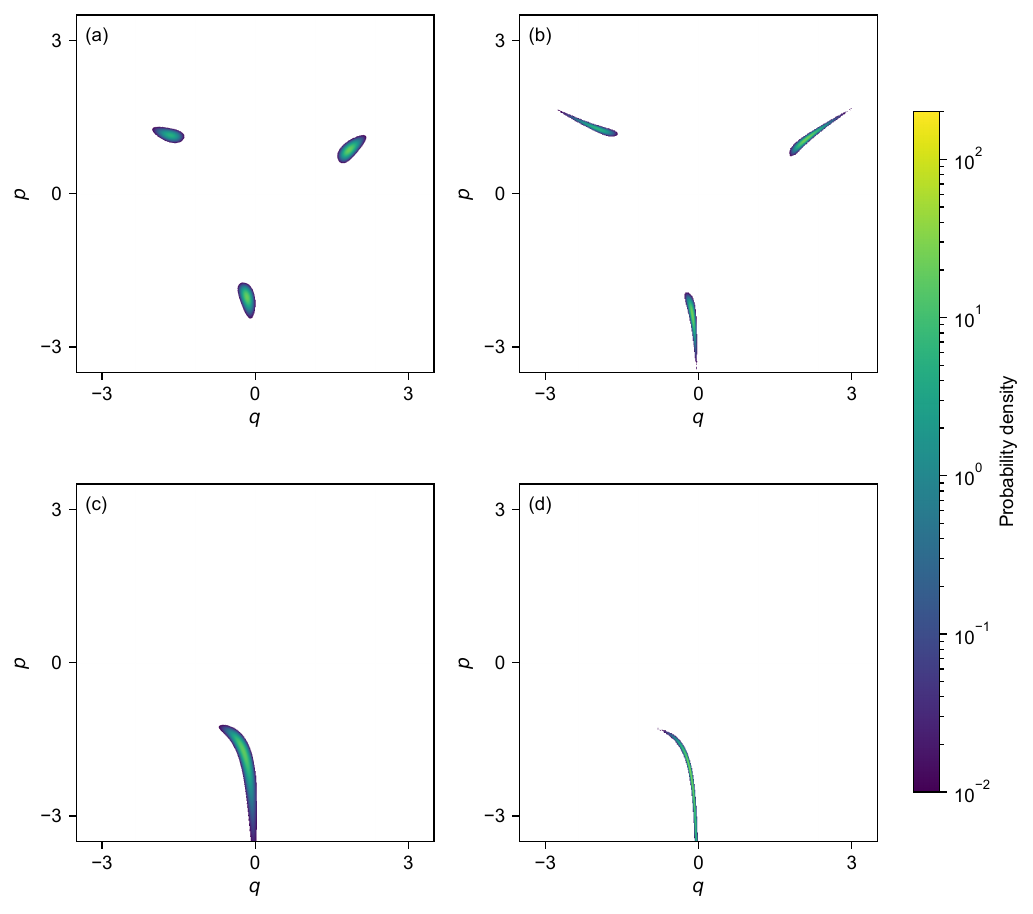}
 \caption{Quantum Husimi density (a, c) and classical transport with coherent Liouville amplitudes (b, d), for Eq.~\eqref{eq:numerical-initial-density}. The times, rounded, are $0.549$ (a, b) and $1.457$ (c, d); $k=3$, $\hbar=0.0025$, $U=-\mathbf 1_3$, cutoff action $384$. All panels have the same logarithmic scale, blank below $0.01$, without window renormalization. Table~\ref{tab:radial-cutoff-test} illustrates the radial cutoff dependence}
 \label{fig:husimi-liouville}
\end{figure}

In Fig.~\ref{fig:channel-probabilities}, the first-pass weights are $(4/9,1/9,4/9)$ for both coherent and Markov propagation; after the second passage they are $(0,0,1)$ and $(8/27,8/27,11/27)$, respectively. Probabilities include the full radial range in angular sectors centered at $2\pi l/3$ with half-width $\pi/3$.

\begin{figure}[H]
 \centering
 \includegraphics[width=\textwidth]{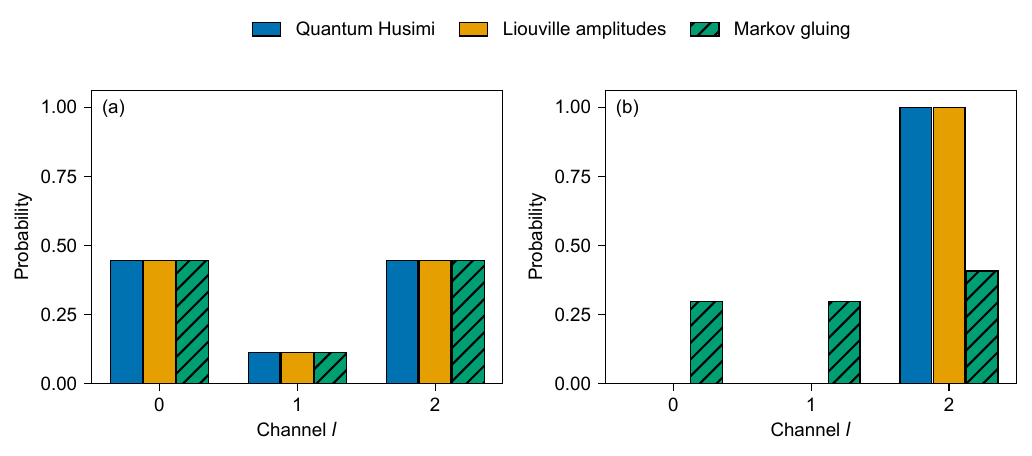}
 \caption{Full-angular-sector probabilities at the first (a) and second (b) times of Fig.~\ref{fig:husimi-liouville}, with unrestricted radial integration. Quantum evolution and coherent Liouville transport are compared with Markov gluing. The limiting first-pass weights coincide; after the second passage coherent weights are $(0,0,1)$, while Markov weights are $(8/27,8/27,11/27)$}
 \label{fig:channel-probabilities}
\end{figure}

Spatial discrepancies include a radial cutoff dependence (Table~\ref{tab:radial-cutoff-test}); classical transport also need not reproduce finite-$\hbar$ Husimi widths. No plotting-window renormalization is applied: the second classical panel contains about $86.8\%$ of the mass. Residual classical channel weights reflect dispersion of passage times. The figures illustrate splitting and recombination, not a full spatial-error bound.

\begin{table}[H]
 \centering
 \caption{Median action of the full quantum Husimi measure for even sector dimension $N$, at $\hbar=0.0025$ and the two times in Fig.~\ref{fig:husimi-liouville}. The central classical action at the first time is approximately $2.787$, from Eq.~\eqref{eq:HamEqs}. Action values are rounded to three decimal places}
 \label{tab:radial-cutoff-test}
 \begin{tabular}{rrrr}
  \hline
  \shortstack{Sites per\\sector $N$} & \shortstack{Cutoff action\\$3\hbar N$} & \shortstack{Median action\\$t\approx0.549$} & \shortstack{Median action\\$t\approx1.457$}\\
  \hline
  $3200$ & $24$ & $1.063$ & $0.817$\\
  $12800$ & $96$ & $1.580$ & $1.070$\\
  $51200$ & $384$ & $2.050$ & $1.595$\\
  $102400$ & $768$ & $2.231$ & $1.849$\\
  $204800$ & $1536$ & $2.376$ & $2.072$\\
  $409600$ & $3072$ & $2.487$ & $2.258$\\
  $819200$ & $6144$ & $2.570$ & $2.405$\\
  $1638400$ & $12288$ & $2.632$ & $2.518$\\
  $3276800$ & $24576$ & $2.677$ & $2.603$\\
  \hline
 \end{tabular}
\end{table}

For the medians in Table~\ref{tab:radial-cutoff-test}, cubic extrapolation in $N^{-1/2}$ uses a leading correction of order $N^{-1/2}$. Extrapolation from the four preceding dimensions reproduces the medians at $N=3276800$ with discrepancies of approximately $4.9\cdot10^{-6}$ and $6.2\cdot10^{-5}$ at the two observation times, respectively.

\section*{Conclusion}

Higher-order squeezing makes the need for a boundary rule visible in both descriptions. The classical orbit escapes in finite time, while the quantum expression has multiple self-adjoint realizations. The ray matrix links the two choices: monomial matrices give deterministic gluings, and general unitary matrices split a localized packet into classical branches.

After the first passage through infinity, the Husimi measure concentrates on the outgoing classical branches with weights $|(S_{\mathrm q}(U))_{lj}|^2$. At positive energy, Corollary~\ref{cor:coherent-returns} also follows later returns: branches in distinct classical sectors retain relative phases that affect the outgoing probabilities. Spatial separation alone does not turn their superposition into an incoherent mixture. The limiting Husimi measure records the branch weights, not those phases, so a one-pass stochastic rule does not describe all subsequent evolution.

Complex Liouville amplitudes retain the information needed for these returns within the extension framework of Refs.~\cite[Remark~6 and Theorems~19--20]{Sakbaev2025} and~\cite[Remark~2 and Sec.~5]{SakbaevLob2026}. This permits a hybrid quantum--classical interpretation: classical transport in $(c,s)$ accompanied by coherent finite-dimensional channel dynamics. For related wavefunction formulations of hybrid dynamics, see Ref.~\cite[Sec.~3]{BondarGayBalmazTronci2019}. Here channels label classical branches, not an independent quantum subsystem. Our results concern Husimi concentration limits, not convergence on a full hybrid observable algebra or unitary equivalence of the full quantum and Liouville groups.

Higher-order squeezing Hamiltonians can also arise in the parametric approximation of essentially self-adjoint multimode Hamiltonians, where pump-mode operators are replaced by complex amplitudes. This approximation is related to first-order perturbation theory for systems in nonequilibrium reservoirs~\cite{KarasevTeretenkov2025}, connecting such models to open-quantum-system methods. It can, however, omit nonperturbative effects such as those considered in Ref.~\cite{SinghTeretenkov2026}.

\section*{Acknowledgments}
The author is grateful to V.~Zh.~Sakbaev for many fruitful discussions, careful reading of intermediate versions of the manuscript, and valuable comments and ideas that contributed to the formulation of the problem and the development of this work.
The author also thanks R.~Singh, A.~V.~Utkin, and V.~I.~Yashin for fruitful discussions of the questions considered in this work.
The author is grateful to the referee for valuable comments that substantially improved the manuscript.

This work was supported by the Russian Science Foundation (grant no.~25-21-00700, \url{https://rscf.ru/en/project/25-21-00700/}).

\end{document}